\RequirePackage{fix-cm}
\documentclass[smallextended]{svjour3}

\usepackage{graphicx}
\RequirePackage{color}
 \usepackage[dvipsnames]{xcolor}
\usepackage{tikz}
\usepackage{booktabs}
\usepackage{tabularx}

\usepackage{amsmath}
\numberwithin{equation}{section}
\usepackage{amssymb}
\usepackage{bm}
\usepackage{mathrsfs}
\usepackage{tensor}
\usepackage{physics}
\usepackage{accents}
\usepackage{cite}
\usepackage{multirow}

\RequirePackage[
    colorlinks=true,
    urlcolor=blue,
    anchorcolor=blue,
    citecolor=blue,
    filecolor=blue,
    linkcolor=blue,
    menucolor=blue,
    linktocpage=true,
    pdfproducer=medialab,
    pdfa=true
]{hyperref}

\hypersetup{
  unicode,
  pdfauthor={Robbert W. Scholtens, Dave B. H. Verweg},
  pdftitle={Gauge invariant frames of Bianchi models in separated form},
  pdfsubject={General Relativity and Gravitation preprint},
  pdfkeywords={Bianchi cosmologies, spatial homogeneity, gauge invariance, invariant frames, Killing vector fields, gravitational waves},
  pdfproducer={LaTeX},
  pdfcreator={pdflatex}
}

\newcommand{\mc}{\mathcal}

\newcommand{\Lie}{\mathcal{L}}
\renewcommand{\dd}{\mathrm{d}}
\newcommand{\Diff}{\operatorname{Diff}}
\newcommand{\GL}{\operatorname{GL}}

\renewcommand{\qand}{\quad\text{and}\quad}

\newcommand{\cs}[3]{\tensor{\Gamma}{^{#1}_{#2}_{#3}}}
\newcommand{\cshat}[3]{\tensor{\Gamma}{^{\hat{#1}}_{\hat{#2}}_{\hat{#3}}}}
\newcommand{\ric}[2]{\tensor{R}{_{#1}_{#2}}}
\newcommand{\richat}[2]{\tensor*{R}{_{\hat{#1}}_{\hat{#2}}}}

\journalname{Journal} 

\begin{document}
\title{Gauge Invariant Frames of Bianchi Models in Separated Form}
\author{Dave B. H. Verweg\footnote{Corresponding author} \and Robbert W. Scholtens}

\institute{
    D. B. H. Verweg \at
    Kapteyn Astronomical Institute, University of Groningen, The Netherlands 
    \\
    \and
    R. W. Scholtens \at
    Kapteyn Astronomical Institute, University of Groningen, The Netherlands\\
    Bernoulli Institute for Mathematics, Computer Science and Artificial Intelligence, University of Groningen, The Netherlands 
    \\
    \email{
    \href{mailto:dave@astro.verweg.net}{dave@astro.verweg.net},
    \href{mailto:r.w.scholtens@rug.nl}{r.w.scholtens@rug.nl}
    }
}
\titlerunning{Gauge Invariant Frames of Bianchi Models}
\authorrunning{D. B. H. Verweg \and R. W. Scholtens}

\date{Received: date / Accepted: date}
\maketitle
\begin{abstract}
    Bianchi spacetimes are usually written in \emph{separated form}: metric coefficients depend only on cosmic time, while spatial dependence is carried by an invariant coframe. In this representation, geometric quantities and the Einstein equations reduce to a time-dependent system. In practice, however, observer- and gravitational-wave-adapted gauges mix time and space (e.g., from integrability-adapted to synchronous comoving coordinates). Homogeneity is preserved, but separability is no longer manifest, which complicates computation. To address this, we give a gauge-invariant reconstruction of separated frames using only the spatial Killing algebra. Given a basis $\{\xi_a\}$, we construct an invariant spatial frame $\{X_i\}$ from $[\xi_a,X_i]=0$, then dualize to $\{e^i\}$. The metric then takes the canonical form $g=\gamma_{\mu\nu}(t)\,e^\mu e^\nu$ in any gauge. We prove diffeomorphism covariance of the construction, and demonstrate computational gain in a Bianchi II example. The framework recovers canonical adapted-gauge constructions (in the sense of Jantzen) while remaining stable under arbitrary $\mathrm{Diff}(\mathcal M)$ gauge transformations, and is directly suited to gravitational-wave Bianchi models and related exact solutions of recent interest of which we give a example of the Bianchi VII type.
\end{abstract}

\section{Introduction}\label{sec:intro}

\subsection{Separated form and the classical computational gain}

Spatially homogeneous cosmologies admit a rigid symmetry structure: a three-dimensional Lie algebra
of Killing vector fields acts simply transitively on spatial hypersurfaces, and the corresponding Bianchi type is encoded in its structure constants \cite{Jantzen1979Bianchi,Taub1951,RyanShepley1975,EllisMaartensMacCallum2012,PapadopoulosGrammenos2012JMP}. This symmetry is classically exploited by writing the metric in a
\emph{separated form}
\begin{equation}\label{eq:intro:separated}
g=\gamma_{\mu\nu}(t)\,\omega^\mu\,\omega^\nu,
\end{equation}
where ${\omega^i}$ is an invariant (Maurer--Cartan) spatial coframe and $\gamma_{\mu\nu}$ depends only on cosmic time \cite{Jantzen1979Bianchi,RyanShepley1975,EllisMaartensMacCallum2012}. Equivalently, in many classical coordinate presentations this structure is encoded by allowing the coordinate components of the metric to have a prescribed spatial dependence, fixed by the Bianchi type, so that the corresponding Killing fields realize the desired Lie algebra. In an invariant coframe, the same spatial dependence is carried instead by the coframe and its structure constants, while the metric coefficients themselves are functions of time only. In this representation, Einstein's equations reduce to an ODE system, and many geometric quantities simplify accordingly \cite{Jantzen1979Bianchi,RyanShepley1975}. The central computational advantage is that the metric coefficients $\gamma_{\mu\nu}(t)$ have no spatial derivatives. Thus connection coefficients, curvature components, ADM quantities, tidal tensors, and geodesic deviation equations can be expressed in terms of time derivatives of $\gamma_{\mu\nu}(t)$ together with the algebraic structure constants of the spatial frame.

Canonical, coordinate-tied constructions of such invariant coframes, and the associated ``separation gauge'', go back to the classic analysis of Jantzen and its descendants, where one identifies each homogeneous slice with a group manifold and parametrises the invariant forms in a fixed group chart \cite{Jantzen1979Bianchi,Christodoulakis2001JMP,ChristodoulakisTerzis2006JMP}. The resulting formalism is extremely effective in the canonical gauge, precisely because the separation is built in at the level of the chosen coframe and chart. The main question of this paper is how to retain this ODE-level computational gain when the metric is presented in a gauge whose coordinate components are not adapted to this canonical separated representation. Our aim is to answer this question in a computationally explicit and practical manner, keeping the entry barrier as low as possible for specialists and for practitioners without an extensive background in differential geometry alike (see also Sect.~\ref{subsec:scope}).

\subsection{The problem: separation hidden by non-adapted gauges}

In applications, one often works in gauges that are not adapted to the canonical Bianchi presentation \cite{CoussaertHenneaux1993CQG}. Here, and throughout the paper, we use ``gauge'' in the diffeomorphism sense~\cite{trautman1980fiber,Jost2017Riemannian} (cf.\ the discussion in \cite[Sect.~3]{Verweg2024AveragingFoliations}): a gauge is a choice of representative or presentation of the spacetime metric within its diffeomorphism class, and a gauge transformation is a spacetime diffeomorphism acting on the metric and on the auxiliary structures used to represent it, such as the frame, coframe, and time function. This tension is particularly sharp in gravitational-wave settings on anisotropic backgrounds. One may begin in coordinates selected by exact integrability properties, for example Hamilton--Jacobi separability, explicit integration of geodesic deviation, or closed-form expressions for tidal observables, and then pass to synchronous coordinates comoving with a freely falling observer to extract physical observables in a laboratory frame \cite{Osetrin2006_stackel_separation,Stackel1893,OsetrinOsetrinOsetrina2022}. In Shapovalov-type and related exact-wave constructions, such transformations typically introduce explicit time--space cross terms and coordinate-dependent mixing \cite{ShapovalovBianchiIII_2306.10605,OsetrinOsetrinOsetrina2022,Osetrin2006_stackel_separation}.

Spatial homogeneity is unchanged, but coordinate-level separation is obscured: metric components acquire explicit spatial dependence, and the ODE structure that was manifest in a separated coframe becomes hidden behind PDE-heavy coordinate expressions \cite{OsetrinOsetrinOsetrina2022,ShapovalovBianchiIII_2306.10605}. This is precisely the ``privileged-to-synchronous'' paradigm emphasised in the Shapovalov gravitational-wave literature \cite{ShapovalovBianchiIII_2306.10605,OsetrinOsetrinOsetrina2022,Osetrin2006_stackel_separation}, and it is precisely this operational bottleneck, the loss of manifest ODE structure in non-adapted gauges, that motivates the present framework.

The point can be summarized schematically as follows (see Table~\ref{tab:comparison}). In a privileged or canonical gauge, the separated Bianchi structure is manifest in the chosen presentation. In a synchronous coordinate gauge, the same homogeneous geometry can look spatially dependent at the component level. Our framework restores the separated representation in the same gauge by replacing the coordinate basis with a reconstructed invariant coframe.

\begin{table}[ht]
    \centering
    \caption{Comparison of metric representations for the Bianchi VII spacetime.}
    \label{tab:comparison}
    \begin{tabularx}{\textwidth}{@{} l XXX @{}}
    \toprule
    \textbf{Feature} & \textbf{Privileged Gauge} & \textbf{Synchronous Gauge} & \textbf{Our Framework} \\
    \midrule
    \textbf{Dependence} & Purely temporal ($t$) & Mixed ($\tau, \tilde{x}^k$) & \text{Purely temporal} ($\tilde{t}$) \\
    \textbf{Complexity} & ODE manifest & PDE-coordinate form & \text{Restored ODE form} \\
    \textbf{Basis} & Coordinate $dx^\alpha$ & Coordinate $d\tilde{x}^\alpha$ & \text{Invariant coframe} $\tilde{e}^\mu$ \\
    \bottomrule
    \end{tabularx}
\end{table}

\subsection{Gauge-invariant reconstruction from the Killing algebra}

The central observation underlying the present work is that \eqref{eq:intro:separated} is a property of a \emph{coframe adapted to the homogeneous foliation}, not of a particular coordinate chart \cite{Jantzen1979Bianchi,RyanShepley1975,EllisMaartensMacCallum2012,Lee2010}. Diffeomorphisms can hide separation in components, but they do not remove the existence of an adapted separated coframe. Accordingly, the right strategy after a gauge move is not to re-identify a canonical group chart and redo a coordinate-tied construction, but to reconstruct an invariant separated coframe directly from invariant data, namely from the spatial Killing algebra itself \cite{Jantzen1979Bianchi,Christodoulakis2001JMP,PapadopoulosGrammenos2012JMP}.

Concretely, given a Killing triple $\{\xi_i\}$ generating the Bianchi homogeneity subgroup, one reconstructs an invariant spatial frame $\{X_j\}$ by solving
\begin{equation}\label{eq:intro:commuting}
    [\xi_i,X_j]=0,
    \qquad i,j=1,2,3,
\end{equation}
and then takes the dual coframe $\{e^\mu\}$ to obtain a separated representation
\begin{equation}\label{eq:intro:reconstructed}
    g=\gamma_{\mu\nu}(t)\,e^\mu e^\nu
\end{equation}
in the \emph{same gauge} in which the metric is presented.

This reconstruction principle, together with the proof that separation is stable under diffeomorphisms, forms the core of our gauge-invariant framework \cite{Jantzen1979Bianchi,Christodoulakis2001JMP,Lee2010,CoussaertHenneaux1993CQG}.

\subsection{Advantage and scope of the paper}\label{subsec:scope}
The purpose of this paper is therefore to provide a gauge-invariant reconstruction of separated frames for Bianchi spacetimes from the spatial Killing algebra. In the canonical gauge, our construction is equivalent to the classical coordinate-tied methods~\cite{Jantzen1979Bianchi,Christodoulakis2001JMP}. Its advantage is that it remains valid in arbitrary gauges, including those that arise naturally in gravitational-wave applications where one moves between privileged integrable coordinates and synchronous laboratory frames~\cite{ShapovalovBianchiIII_2306.10605,OsetrinOsetrinOsetrina2022,Osetrin2006_stackel_separation}.

This advantage is not only formal, but computational. In the reconstructed coframe, the metric coefficients are $\gamma_{\mu\nu}(t)$ and therefore have no spatial derivatives. In a general non-coordinate frame, the Levi--Civita connection is computed from the Koszul formula, so the only derivatives of the metric coefficients are ordinary time derivatives along $X_0=\partial_t$. All remaining spatial dependence enters algebraically through the time-independent structure functions of the reconstructed frame. Consequently, the connection, curvature components, ADM quantities, tidal tensors, and geodesic deviation equations are built from $\gamma_{\mu\nu}(t)$, its time derivatives, and the Bianchi structure coefficients. This restores the ODE-level computational structure familiar from canonical Bianchi gauges, even when the coordinate components of the metric in a synchronous or observer-adapted gauge are manifestly non-separated. The gain is anticipated schematically in Table~\ref{tab:comparison} and demonstrated explicitly in the examples.

There is a corresponding trade-off. The reconstructed frame is generally a point-dependent, non-coordinate frame, so its individual components need not have the same immediate coordinate or observer interpretation as the components in the original chart. We therefore do not propose to replace physically motivated coordinates. Rather, the framework separates two roles: the chosen gauge may remain the natural one for interpretation, while the reconstructed invariant frame supplies the efficient representation for geometric computation.

The examples are chosen to make this point immediate. Section~\ref{sec:examples} first gives a Bianchi II example with an explicit time--space cross term and shows, in a controlled setting, how the reconstructed invariant coframe restores separation and reduces the apparent coordinate complexity \cite{Valent2009GRG}. A second gravitational-wave motivated example, drawn from a Shapovalov-type construction with Bianchi VII homogeneity and its privileged-to-synchronous transformation, shows the stronger physical implication: one may keep the synchronous observer-adapted gauge while recovering the separated ODE character of the homogeneous sector in the reconstructed frame \cite{ShapovalovBianchiIII_2306.10605,OsetrinOsetrinOsetrina2022,Osetrin2006_stackel_separation}. For related motivations from anisotropic perturbation and wave studies on cosmological backgrounds, see also \cite{GuhaDatta2020IJMPD,BianchiI_GW_2205.05733,LudwickWilliams2025PhLB}.

The paper is organized as follows. Section~\ref{sec:sepform} formulates separated form in coordinate-free terms and proves diffeomorphism stability \cite{Lee2010}. Section~\ref{sec:reconstruction} presents the Killing-algebra reconstruction and establishes its naturality \cite{PapadopoulosGrammenos2012JMP,Scholtens2024Bianchi}. Section~\ref{sec:examples} gives the Bianchi II model with an explicit time--space cross term \cite{Valent2009GRG}. The gravitational-wave Bianchi VII example then demonstrates the same reconstruction in a privileged-to-synchronous setting \cite{ShapovalovBianchiIII_2306.10605,OsetrinOsetrinOsetrina2022,Osetrin2006_stackel_separation}. Section~\ref{sec:conclusion} concludes.

We use signature $(-,+,+,+)$. Greek indices $\mu,\nu,\dots$ run over $0,1,2,3$, and Latin indices $i,j,\dots$
over $1,2,3$. A \emph{frame} ${X_\mu}$ denotes a basis of vector fields and a \emph{coframe} ${e^\mu}$ its
dual basis, $e^\mu(X_\nu)=\delta^\mu{}_\nu$. A \emph{gauge transformation} is a diffeomorphism
$f\in\mathrm{Diff}(\mathcal M)$ acting by pushforward/pullback; we track the transformation of the time
function explicitly to maintain a foliation-adapted geometric description \cite{Lee2010,Verweg2024AveragingFoliations}.

\section{Separated form and its gauge invariance}\label{sec:sepform}

\subsection{Separated form as a geometric definition}\label{subsec:separatedform}

Let $(\mc{M},g)$ be a simply-connected Lorentzian $4$-manifold. Fix a smooth frame $\{X_\mu\}_{\mu=0}^3$ and we denote its coframe by $\{e^\mu\}_{\mu=0}^3$, not assumed to be coordinate-induced. Our aim is to isolate, in a coordinate-free manner, the precise structural input that underlies the familiar ODE reduction of spatially homogeneous dynamics in a
left-invariant basis \cite{Jantzen1979Bianchi,RyanShepley1975,Lee2010}.

\begin{definition}\label{def:separated}
    We say that $g$ is in \textbf{separated form} with respect to $(X_\mu,e^\mu)$ if there exists a smooth time function
    $t:\mc{M}\to\mathbb{R}$ such that
    \begin{equation}\label{eq:sep:def:e0}
    e^0=\dd t \qand X_0=\partial_t,
    \end{equation}
    and the metric is decomposed as
    \begin{equation}\label{eq:sep:def}
    g=\gamma_{\mu\nu}(t)\,e^\mu e^\nu,
    \end{equation}
    where the coefficients $\gamma_{\mu\nu}$ depend only on $t$, together with the coframe invariance condition
    \begin{equation}\label{eq:sep:Lie}
    \Lie_{X_0} e^\mu = 0,\qquad \mu=0,1,2,3.
    \end{equation}
\end{definition}

\vspace{3mm}
\noindent
In cosmology, spacetimes are standardly assumed to be globally hyperbolic. For such spacetimes, we are thus always able to construct a cosmic time $t$, and hence any such spacetime metric can be written in a separated form. In general, besides standard cosmology, a spacetime metric can be written into a separated form if its topology allows the existence of a global cosmic time function $t$.

Definition~\ref{def:separated} packages the separation property as a statement about an \emph{adapted coframe}:
all time dependence resides in $\gamma_{\mu\nu}(t)$, while all spatial structure is carried by the one-forms
$e^i$ and their (generally nontrivial) commutation relations. In particular, introduce the structure functions
$C^\alpha{}_{\mu\nu}$ of the frame by
\begin{equation}\label{eq:sep:structurefunctions}
[X_\mu,X_\nu]=C^\alpha{}_{\mu\nu}\,X_\alpha,
\qquad
\dd e^\alpha=-\frac12\,C^\alpha{}_{\mu\nu}\,e^\mu\wedge e^\nu.
\end{equation}
Then \eqref{eq:sep:Lie} implies that the structure functions are time-independent:
\begin{equation}\label{eq:sep:C_timeindep}
0=\Lie_{X_0}(\dd e^\alpha)=\dd(\Lie_{X_0}e^\alpha)
\;\;\Longrightarrow\;\;
\Lie_{X_0}C^\alpha{}_{\mu\nu}=0.
\end{equation}
Consequently, any spatial dependence enters only through $C^\alpha{}_{\mu\nu}(x)$ (equivalently through the
noncommutativity of the chosen frame), and not through spatial derivatives of the metric coefficients
$\gamma_{\mu\nu}$. This is the invariant content behind the standard statement that, in a symmetry-adapted
basis, the homogeneous Einstein system closes as an ODE system \cite{Jantzen1979Bianchi,RyanShepley1975,Lee2010,EllisMaartensMacCallum2012}.

\subsection{Diffeomorphism stability}

Separated form is formulated entirely in terms of \textit{tensorial} data: a metric $g$, a time function, and an adapted frame--coframe pair. It is therefore expected to behave naturally under gauge transformations, i.e.\ diffeomorphisms. The following proposition makes this precise. Conceptually, it explains why coordinate-level loss of separation under a gauge move is not an obstruction -- the separated structure is transported covariantly once one transforms both the coframe and the time function.

\begin{proposition}[Diffeomorphism stability]\label{prop:diffeostable}
    Let $f:\mc{M}\to\mc{N}$ be a diffeomorphism and assume that $g$ is in separated form on $\mc{M}$ with respect to $(X_\mu,e^\mu)$ and time function $t$. Define the pushed-forward time function
    \begin{equation}\label{eq:sep:push_time}
    \tilde t:=t\circ f^{-1}.
    \end{equation}
    Then the pushforward metric $f_*g$ is in separated form on $\mc{N}$ with respect to the pushed-forward frame--coframe pair $(f_*X_\mu,f_*e^\mu)$ and time function $\tilde t$.
\end{proposition}

\vspace{2mm}
\noindent
In gravitational-wave Bianchi constructions it is common to work first in privileged (integrability-adapted) coordinates and later transform to synchronous coordinates adapted to a freely falling observer
\cite{OsetrinOsetrinOsetrina2022,ShapovalovBianchiIII_2306.10605,Osetrin2006_stackel_separation}. Proposition~\ref{prop:diffeostable} isolates the geometric
content: separated form persists at the level of an adapted coframe, provided the time function is transformed
consistently with the diffeomorphism. This is precisely the mechanism by which a gauge transformation may
obscure separation in coordinate components while leaving the underlying separated structure intact.

\subsubsection{Proof of diffeomorphism stability}
First, observe
\begin{equation}
(f_*e^\mu)(f_*X_\nu)=e^\mu(X_\nu)=\delta^\mu{}_\nu, 
\end{equation}
so $(f_*X_\mu,f_*e^\mu)$ is dual. Since $e^0=\dd t$, we have $f_*e^0=\dd(t\circ f^{-1})=\dd\tilde t$. Naturality of the Lie derivative gives $\Lie_{f_*X_0}(f_*e^\mu)=f_*(\Lie_{X_0}e^\mu)=0$. Finally, for $p\in\mc{N}$, \begin{equation} (f_*g)_p(f_*X_\mu,f_*X_\nu)=g_{f^{-1}(p)}(X_\mu,X_\nu)=\gamma_{\mu\nu}(t(f^{-1}(p)))=\gamma_{\mu\nu}(\tilde t(p)), \end{equation} which is the separated representation.

\section{Reconstruction of invariant frames from the Killing algebra}
\label{sec:reconstruction}


\subsection{Set-up and physical idea}

Section~\ref{sec:intro} identified the main limitation of the standard Bianchi
construction: it is most efficient when one has already chosen coordinates adapted to
the homogeneous spatial symmetry. In such a canonical gauge, the invariant
Maurer--Cartan coframe is built into the coordinate description, and the metric takes
separated form from the start. After a general gauge transformation, however, and in
particular after a transformation to synchronous observer coordinates, this adapted
coordinate description may be lost.

The purpose of this section is to show that the adapted coframe itself is not lost. It
can be reconstructed directly from the spatial Killing algebra. This is the key
advantage of the framework. We do not reconstruct a preferred coordinate system, and
we do not reparametrise the spatial group manifold. Instead, we reconstruct the
invariant spatial frame whose dual coframe carries the Bianchi structure. Once this
frame is found, the metric coefficients are again functions of the foliation time only,
as in the separated form of Section~\ref{sec:sepform}. Related constructions of
invariant frames from Killing data appear in
\cite{PapadopoulosGrammenos2012JMP}, and in the algorithmic construction used in~\cite{Scholtens2024Bianchi}.

Let $(\mathcal M,g)$ be a spatially homogeneous spacetime. Let $\mathcal K$ be the
three-dimensional Lie algebra of spatial Killing fields spanning the tangent spaces of
the homogeneous hypersurfaces. Fix a basis $\{\xi_a\}_{a=1}^3\subset\mathcal K$
satisfying
\begin{equation}
\label{eq:xi:structure}
[\xi_a,\xi_b]=C^c{}_{ab}\,\xi_c,
\end{equation}
where the constants $C^c{}_{ab}$ encode the Bianchi type
\cite{Jantzen1979Bianchi,RyanShepley1975,EllisMaartensMacCallum2012}. The object
to reconstruct is not a coordinate chart, but a \textit{spatial frame} which is
constant along the group orbits.

\begin{definition}
\label{def:invframe}
A spatial frame $\{X_i\}_{i=1}^3$ is \textbf{invariant} with respect to the Killing
basis $\{\xi_a\}$ if
\begin{equation}
\label{eq:inv:commute}
[\xi_a,X_i]=0,
\qquad a,i=1,2,3.
\end{equation}
\end{definition}

\noindent
The condition \eqref{eq:inv:commute} is the coordinate-free replacement for using
left- or right-invariant frames in a canonical Bianchi chart. It says that the frame
fields $X_i$ are transported unchanged along the spatial isometry orbits. This is
precisely the structure that remains meaningful after a gauge transformation, even
when the coordinate components of the metric are no longer separated
\cite{Jantzen1979Bianchi,Christodoulakis2001JMP,ChristodoulakisTerzis2006JMP,
PapadopoulosGrammenos2012JMP}.

\subsection{Invariant-frame reconstruction}

The first step is purely algebraic and differential: from the Killing basis
$\{\xi_a\}$ we reconstruct a spatial frame satisfying \eqref{eq:inv:commute}. This is
the part of the construction that replaces the choice of a canonical group chart.

\begin{theorem}[Existence and uniqueness of the invariant frame]
\label{thm:invframe:existence}
Let $(\mathcal M,g)$ be spatially homogeneous, with spatial Killing basis
$\{\xi_a\}_{a=1}^3$ satisfying \eqref{eq:xi:structure}. On each homogeneous
hypersurface, there exists locally a spatial frame $\{X_i\}_{i=1}^3$ satisfying
\begin{equation}
\label{eq:invframe:theorem}
[\xi_a,X_i]=0.
\end{equation}
After fixing a reference point $p$ on the hypersurface and a normalisation
$X_i(p)=M_i{}^a\,\xi_a(p)$ with $M\in\GL(3)$, the local solution is unique.
Under the usual global assumptions for homogeneous Bianchi slices, for instance when
the homogeneous hypersurface is simply connected and the Killing action is complete,
the construction extends globally on the hypersurface.
\end{theorem}

\begin{proof}
Write the desired invariant frame as
\begin{equation}
\label{eq:Xi_expand}
X_i=X_i{}^b\,\xi_b .
\end{equation}
Using \eqref{eq:xi:structure}, the condition $[\xi_a,X_i]=0$ becomes
\begin{equation}
\label{eq:reconstruct:linearPDE}
\xi_a\!\left(X_i{}^b\right)
=
-C^b{}_{ac}X_i{}^c,
\qquad a,b,i=1,2,3.
\end{equation}
Indeed,
\begin{equation}
[\xi_a,X_i]
=
\left(\xi_a(X_i{}^b)+X_i{}^cC^b{}_{ac}\right)\xi_b .
\end{equation}
Equation \eqref{eq:reconstruct:linearPDE} is the reconstruction equation.

This is a linear first-order system along the homogeneous orbit. To see the local
existence explicitly, restrict \eqref{eq:reconstruct:linearPDE} to a smooth curve
$x(s)$ tangent to the orbit,
\begin{equation}
\dot x(s)=F^a(s)\xi_a\big|_{x(s)} .
\end{equation}
Then \eqref{eq:reconstruct:linearPDE} becomes the linear matrix ODE
\begin{equation}
\label{eq:invframe_flow}
\frac{d}{ds}X_i{}^b(s)
=
-F^a(s)C^b{}_{ac}X_i{}^c(s).
\end{equation}
Since the coefficients are smooth, local existence and uniqueness along the curve
follow from standard ODE theory \cite{Lee2010}.

The system is compatible on the orbit because the matrices
$(\mathbf C_a)^b{}_{c}:=C^b{}_{ac}$ furnish the adjoint representation of the
Killing algebra. Equivalently, if
\begin{equation}
D_aV^b:=\xi_a(V^b)+C^b{}_{ac}V^c ,
\end{equation}
then \eqref{eq:reconstruct:linearPDE} is $D_aX_i{}^b=0$, and the Jacobi identity
implies
\begin{equation}
[D_a,D_b]=C^c{}_{ab}D_c .
\end{equation}
Thus the equations are Frobenius-compatible. Consequently the local solution can be
continued consistently on the homogeneous hypersurface. With a fixed normalisation
\begin{equation}
\label{eq:init_cond}
X_i(p)=M_i{}^a\,\xi_a(p),
\qquad M\in\GL(3),
\end{equation}
the linear system has a unique local solution. If the homogeneous hypersurface is
simply connected and the Killing action is complete, this continuation is path
independent and gives a global invariant frame on the hypersurface.
\end{proof}

\subsection{Generation of separated form}

The next step is geometric. Once the invariant frame has been reconstructed, its dual
coframe gives the separated representation of the metric.

\begin{proposition}[Separated form from the invariant frame]
\label{prop:separated:from:frame}
Let $t$ be a foliation time whose level sets are homogeneous hypersurfaces, and choose
$X_0=\partial_t$ so that
\begin{equation}
[X_0,\xi_a]=0 .
\end{equation}
Let $\{X_i\}$ be the invariant spatial frame from
Theorem~\ref{thm:invframe:existence}, and let $\{e^\mu\}$ be the coframe dual to
$\{X_\mu\}=\{X_0,X_i\}$. Then the metric takes the separated form
\begin{equation}
\label{eq:reconstruct:sep}
g=\gamma_{\mu\nu}(t)\,e^\mu\otimes e^\nu,
\qquad
\gamma_{\mu\nu}(t):=g(X_\mu,X_\nu).
\end{equation}
If the spatial frame is propagated along the foliation time by
\begin{equation}
[X_0,X_i]=0 ,
\end{equation}
then also
\begin{equation}
\Lie_{X_0}e^\mu=0 .
\end{equation}
\end{proposition}

\begin{proof}
Define
\begin{equation}
\label{eq:gamma_def}
\gamma_{\mu\nu}:=g(X_\mu,X_\nu).
\end{equation}
For each spatial Killing field $\xi_a$,
\begin{equation}
\xi_a\big(g(X_\mu,X_\nu)\big)
=
(\Lie_{\xi_a}g)(X_\mu,X_\nu)
+
g([\xi_a,X_\mu],X_\nu)
+
g(X_\mu,[\xi_a,X_\nu]).
\end{equation}
The first term vanishes because $\xi_a$ is Killing. The remaining two terms vanish
because $[\xi_a,X_i]=0$ by construction and $[\xi_a,X_0]=0$ by the choice of
foliation-adapted time. Therefore
\begin{equation}
\xi_a(\gamma_{\mu\nu})=0 .
\end{equation}
Since the $\xi_a$ span the spatial tangent directions on each homogeneous
hypersurface, $\gamma_{\mu\nu}$ is constant on each such hypersurface. Hence
\begin{equation}
\gamma_{\mu\nu}=\gamma_{\mu\nu}(t),
\end{equation}
and, by duality of the frame and coframe,
\begin{equation}
g=\gamma_{\mu\nu}(t)\,e^\mu\otimes e^\nu .
\end{equation}

Finally, if the spatial frame is propagated along the foliation time so that
$[X_0,X_i]=0$, then $\Lie_{X_0}e^\mu=0$ follows from duality:
\begin{equation}
0
=
\Lie_{X_0}\big(e^\mu(X_\nu)\big)
=
(\Lie_{X_0}e^\mu)(X_\nu)
+
e^\mu([X_0,X_\nu]).
\end{equation}
Since this holds for every $X_\nu$, we obtain $\Lie_{X_0}e^\mu=0$.
\end{proof}

\subsection{Naturality and type preservation}

The final step records the properties that make the construction useful in non-adapted
gauges. The reconstruction preserves the Bianchi type, behaves naturally under
diffeomorphisms, and is equivariant under changes of Killing basis.

\begin{proposition}[Structural properties of the reconstruction]
\label{prop:reconstruction:properties}
The invariant-frame reconstruction has the following properties.
\begin{enumerate}
\item \textbf{Type preservation.} The reconstructed spatial frame closes on a Lie
algebra isomorphic to the original Killing algebra. In the normalisation
$X_i(p)=\xi_i(p)$, it carries the opposite structure constants,
\begin{equation}
[X_i,X_j]=-C^k{}_{ij}X_k
\end{equation}
at the reference point, and hence the same Bianchi type.

\item \textbf{Diffeomorphism covariance.} If $f\in\Diff(\mathcal M)$, then the
reconstruction applied to $(f_*g,f_*\xi_a)$ yields the pushed-forward frame
$f_*X_i$ and coframe $f_*e^i$, with the pushed-forward normalisation.

\item \textbf{General linear equivariance.} If the Killing basis is changed by
$\xi'_a=A_a{}^b\xi_b$, with constant $A\in\GL(3)$, then the reconstructed frame
may be chosen as
\begin{equation}
X'_i=(A^{-1})^b{}_{i}X_b .
\end{equation}
\end{enumerate}
\end{proposition}

\begin{proof}
We first prove type preservation. Using $X_i=X_i{}^a\xi_a$,
\begin{equation}
\label{eq:type_preservation}
[X_i,X_j]
=
X_i{}^aX_j{}^b[\xi_a,\xi_b]
+
X_i{}^a\xi_a(X_j{}^c)\xi_c
-
X_j{}^a\xi_a(X_i{}^c)\xi_c .
\end{equation}
Substitution of the reconstruction equation
\eqref{eq:reconstruct:linearPDE} gives
\begin{equation}
\label{eq:type_preservation_simplified}
[X_i,X_j]
=
-\left(X_i{}^aX_j{}^bC^c{}_{ab}(X^{-1})_c{}^k\right)X_k .
\end{equation}
Thus the reconstructed frame closes on the opposite algebra, up to the chosen
normalisation of the frame. Since the opposite of a Lie algebra is canonically
isomorphic to the original algebra, and since a different normalisation only conjugates
the structure constants, the Bianchi type is preserved
\cite{Jantzen1979Bianchi,RyanShepley1975}.

For diffeomorphism covariance, let $f\in\Diff(\mathcal M)$. Then
\begin{equation}
[f_*\xi_a,f_*X_i]=f_*[\xi_a,X_i]=0 .
\end{equation}
Thus $f_*X_i$ satisfies the reconstruction equation for the pushed-forward Killing
basis $f_*\xi_a$. Since the reference normalisation is also pushed forward, uniqueness
from Theorem~\ref{thm:invframe:existence} implies that the reconstructed frame in the
transformed gauge is precisely $f_*X_i$, with dual coframe $f_*e^i$.

Finally, if $\xi'_a=A_a{}^b\xi_b$ with constant $A\in\GL(3)$, then choosing
\begin{equation}
X'_i=(A^{-1})^b{}_{i}X_b
\end{equation}
gives
\begin{equation}
[\xi'_a,X'_i]
=
A_a{}^b(A^{-1})^c{}_{i}[\xi_b,X_c]
=
0 .
\end{equation}
This proves $\GL(3)$-equivariance.
\end{proof}

\subsection{Physical content of the reconstruction}

The proof has three important consequences for the framework.

First, the reconstruction is a \emph{factorisation}, not a gauge fixing. A gauge fixing
would restrict the allowed coordinate transformations. Here no such restriction is
imposed. If $f\in\Diff(\mathcal M)$, then
\begin{equation}
[f_*\xi_a,f_*X_i]
=
f_*[\xi_a,X_i]
=
0 .
\end{equation}
Thus the pushed-forward frame $f_*X_i$ is invariant with respect to the
pushed-forward Killing basis $f_*\xi_a$. With the reference normalisation also
pushed forward to $f(p)$, uniqueness implies that the reconstruction in the
transformed gauge gives precisely $f_*X_i$ and its dual coframe. This is the precise
sense in which the framework is gauge invariant: it reconstructs the separated
coframe in whatever gauge the metric is presented, rather than restricting one to a
canonical coordinate parametrisation
\cite{Lee2010,Jantzen1979Bianchi,Christodoulakis2001JMP,CoussaertHenneaux1993CQG}.
This is why the framework applies directly to the privileged-to-synchronous
transformations discussed in the introduction
\cite{Osetrin2006_stackel_separation,OsetrinOsetrinOsetrina2022,
ShapovalovBianchiIII_2306.10605}.

Second, the Bianchi type is preserved. The reconstructed frame does not introduce a
new spatial geometry. Its commutator algebra is isomorphic to that of the original
Killing algebra, up to the conventional sign associated with using the frame that
commutes with the chosen Killing action. The spatial dependence has simply been
moved to the coframe and its structure coefficients, where it belongs in a separated
description.

Third, the dynamical variables are the metric coefficients
\begin{equation}
\gamma_{\mu\nu}(t)=g(X_\mu,X_\nu).
\end{equation}
Because these coefficients are constant on each homogeneous hypersurface, they depend
only on the foliation time. This is the mathematical reason for the computational
advantage highlighted in the introduction: after reconstruction, the homogeneous sector
is again described by time-dependent quantities, even if the coordinate components of
the same metric look spatially dependent in a non-adapted gauge.

\section{Examples}\label{sec:examples}

The introduction explained the problem with the standard approach. Classical Bianchi
methods are most effective when the coordinates are already adapted to the homogeneous
spatial symmetry. In that setting the metric is separated from the start, and the
geometry reduces to time-dependent quantities. In applications, however, especially in
gravitational-wave models, one often changes to observer-adapted or synchronous
coordinates. These coordinates are physically useful, but they can hide the separated
Bianchi structure: the metric components acquire explicit spatial dependence even
though the spacetime remains spatially homogeneous.

The examples below show two complementary aspects of the reconstruction. The Bianchi II model isolates the mechanism in a controlled setting: a coordinate expression with explicit coordinate mixing (i.e.\ featuring off-diagonal terms) is transformed, not by changing gauge, but by replacing the coordinate basis with the invariant coframe reconstructed from the Killing algebra. This makes the hidden separated form manifest and directly exposes the computational gain, since the connection and curvature are expressed in terms of time-dependent metric coefficients and algebraic frame structure rather than spatially dependent coordinate components.

The Bianchi VII example then shows that the same idea is not merely a simplification of an artificial mixed metric. It applies in the gravitational-wave setting that motivates the paper, where one passes from privileged integrability-adapted coordinates to synchronous coordinates adapted to freely falling observers. In that setting the coordinate components of the metric become manifestly non-separated, but this spatial dependence is representational rather than physical: it reflects the use of a non-invariant coordinate basis, not a loss of spatial homogeneity. The reconstructed coframe therefore separates two roles that are usually entangled: synchronous coordinates remain available for physical interpretation, while the invariant frame restores the ODE-level Bianchi structure for computation.

\subsection{Bianchi II example: restoring separated form and its computational advantage}
\subsubsection{The metric and its Killing algebra}
In order to demonstrate the utility of the separated formalism, we consider the following metric that, with the implicit coordinate choice, is evidently not in separated form:
\begin{equation}\label{eq:bianchiII_metric}
    \begin{split}
        \dd{s}^2&=-\dd{t}^2+a(t)^2\dd{x}^2+(z^2a(t)^2+b(t)^2)\dd{y}^2+c(t)^2\dd{z}^2\\
        &\qquad-2za(t)^2\dd{x}\dd{y}.
    \end{split}
\end{equation}
In order to utilize our methodology, we first need this metric's KVFs. It can be verified that KVFs $\tensor{\xi}{_a}$, $a=1,2,3$, are given by
\begin{equation}
    \tensor{\xi}{_1}=\partial_x,\quad\tensor{\xi}{_2}=\partial_y,\qand\tensor{\xi}{_3}=\partial_z+y\partial_x.
\end{equation}
These KVFs indeed form a Lie algebra, whose structure constants in this basis are $\tensor{C}{^1_2_3}=-\tensor{C}{^1_3_2}=1$, with the others vanishing---hence, this is a Bianchi II Lie algebra.

\subsubsection{Invariant frame reconstruction}
We then utilize the method from \cite{Scholtens2024Bianchi} in order to find the corresponding invariant frame $\tensor{X}{_i}$, $i=1,2,3$:
\begin{equation}\label{eq:bianchiII_Xbasis}
    \tensor{X}{_1}=\partial_x,\quad\tensor{X}{_2}=\partial_y+z\partial_x,\qand\tensor{X}{_3}=\partial_z.
\end{equation}
It can be verified that, indeed, $[\tensor{\xi}{_a},\tensor{X}{_i}]=0$ for all choices of $a$, $i$, and additionally that $[\tensor{X}{_2},\tensor{X}{_3}]=-\tensor{X}{_1}$ with other brackets vanishing, as was to be expected. It is then a small matter to find the dual of the $X$-basis:
\begin{equation}
    \tensor{e}{^1}=\dd{x}-z\dd{y},\quad\tensor{e}{^2}=\dd{y},\qand\tensor{e}{^3}=\dd{z},
\end{equation}
which may be done, for instance, by matrix inversion in the $(x,y,z)$-basis. Substituting this basis into the original metric \eqref{eq:bianchiII_metric}, and additionally setting $\tensor{e}{^0}=\dd{t}\iff\tensor{X}{_0}=\partial_t$, we find the expression
\begin{equation}\label{eq:bianchiII_metric_v2}
    \dd{s}^2=-(\tensor{e}{^0})^2+a(t)^2(\tensor{e}{^1})^2+b(t)^2(\tensor{e}{^2})^2+c(t)^2(\tensor{e}{^3})^2.
\end{equation}
Note that now the metric is in separated form, as was the aim of the construction in this paper: it can be verified straightforwardly that $\mathcal{L}_{X_0}\tensor{e}{^\mu}=0$ for all $\mu=0,1,2,3$. Coincidentally, the off-diagonal element of \eqref{eq:bianchiII_metric} has also been removed.

\begin{table}
    \centering
    \renewcommand{\arraystretch}{1.5}
    \caption{Christoffel symbols in the coordinate basis (regular digits) and the basis from the separated form (hatted digits), based on the metrics of \eqref{eq:bianchiII_metrics}. The entries in \textcolor{red}{red} indicate that in the separated form, the symmetry in the Christoffel symbols' lower indices is generally broken.}
    \begin{tabularx}{\textwidth}{lllll}
        \toprule
        \multirow{13}{*}{\emph{Coordinate basis}} & \multicolumn{2}{l}{$\cs{0}{1}{1}=a\,\partial_ta$} & \multicolumn{2}{l}{$\cs{3}{0}{3}=\frac{\partial_tc}{c}$} \\
        & \multicolumn{2}{l}{$\cs{0}{1}{2}=-za\,\partial_ta$} & \multicolumn{2}{l}{$\cs{3}{3}{0}=\frac{\partial_t c}{c}$} \\
        & \multicolumn{2}{l}{$\cs{0}{2}{1}=-za\,\partial_ta$} & \multicolumn{2}{l}{$\cs{3}{1}{2}=\frac{a^2}{2c^2}$} \\
        & \multicolumn{2}{l}{$\cs{0}{2}{2}=z^2a\,\partial_ta+b\,\partial_tb$} & \multicolumn{2}{l}{$\cs{3}{2}{1}=\frac{a^2}{2c^2}$} \\
        & \multicolumn{2}{l}{$\cs{0}{3}{3}=c\,\partial_tc$} & \multicolumn{2}{l}{$\cs{3}{2}{2}=-z\frac{a^2}{c^2}$} \\
        & \multicolumn{2}{l}{$\cs{1}{0}{1}=\frac{\partial_ta}{a}$} & \multicolumn{2}{l}{$\cs{2}{0}{2}=\frac{\partial_tb}{b}$} \\
        & \multicolumn{2}{l}{$\cs{1}{1}{0}=\frac{\partial_ta}{a}$} & \multicolumn{2}{l}{$\cs{2}{2}{0}=\frac{\partial_tb}{b}$} \\
        & \multicolumn{2}{l}{$\cs{1}{0}{2}=z(\frac{\partial_tb}{b}-\frac{\partial_ta}{a})$} & \multicolumn{2}{l}{$\cs{2}{1}{3}=-\frac{a^2}{2b^2}$} \\
        & \multicolumn{2}{l}{$\cs{1}{2}{0}=z(\frac{\partial_tb}{b}-\frac{\partial_ta}{a})$} & \multicolumn{2}{l}{$\cs{2}{3}{1}=-\frac{a^2}{2b^2}$} \\
        & \multicolumn{2}{l}{$\cs{1}{2}{3}=\frac{1}{2}$} & \multicolumn{2}{l}{$\cs{2}{2}{3}=z\frac{a^2}{2b^2}$} \\
        & \multicolumn{2}{l}{$\cs{1}{3}{2}=\frac{1}{2}$} & \multicolumn{2}{l}{$\cs{2}{3}{2}=z\frac{a^2}{2b^2}$} \\
        & \multicolumn{4}{l}{$\cs{1}{1}{3}=-z\frac{a^2}{b^2}$} \\
        & \multicolumn{4}{l}{$\cs{1}{3}{1}=-z\frac{a^2}{b^2}$} \\
        \midrule
        \multirow{4}{*}{\emph{Separated form}} & $\cshat{0}{1}{1}=a\,\partial_ta$ & $\cshat{1}{0}{1}=\frac{\partial_ta}{a}$ & $\cshat{2}{0}{2}=\frac{\partial_tb}{b}$ & $\cshat{3}{0}{3}=\frac{\partial_tc}{c}$ \\
        & $\cshat{0}{2}{2}=b\,\partial_tb$ & $\cshat{1}{1}{0}=\frac{\partial_ta}{a}$ & $\cshat{2}{2}{0}=\frac{\partial_tb}{b}$ & $\cshat{3}{3}{0}=\frac{\partial_tc}{c}$ \\
        & $\cshat{0}{3}{3}=c\,\partial_tc$ & \textcolor{red}{$\cshat{1}{3}{2}=\tfrac{1}{2}$} & $\cshat{2}{1}{3}=\frac{a^2}{2b^2}$ & $\cshat{3}{1}{2}=-\frac{a^2}{2b^2}$\\
        & & \textcolor{red}{$\cshat{1}{2}{3}=-\tfrac{1}{2}$} & $\cshat{2}{3}{1}=\frac{a^2}{2b^2}$ & $\cshat{3}{2}{1}=-\frac{a^2}{2b^2}$ \\ \bottomrule
    \end{tabularx}
    \label{tab:christoffels}
    \renewcommand{\arraystretch}{1}
\end{table}

\begin{table}
    \centering
    \renewcommand{\arraystretch}{1.5}
    \caption{Components of the Ricci tensor in the coordinate basis (regular digits) and the basis from the separated form (hatted digits), based on the metrics of \eqref{eq:bianchiII_metrics}.}
    \begin{tabularx}{\textwidth}{ll}
         \toprule
         \multirow{7}{*}{\emph{Coordinate basis}} & $\ric{0}{0}=2\!\left(\frac{\partial_ta}{a}\right)^{\!2}-\frac{\partial_t^2a}{a}+2\!\left(\frac{\partial_tb}{b}\right)^{\!2}-\frac{\partial_t^2b}{b}+2\!\left(\frac{\partial_tc}{c}\right)^{\!2}-\frac{\partial_t^2c}{c}$ \\
         & $\ric{1}{1}=a^2\left[\frac{\partial_t^2a}{a}+2\!\left(\frac{\partial_ta}{a}\right)^{\!2}-\frac{\partial_ta}{a}\frac{\partial_tb}{b}-\frac{\partial_ta}{a}\frac{\partial_tc}{c}-\frac{a^2}{2b^2c^2}\right]$\\
         & $\ric{1}{2}=-za^2\left[\frac{\partial_t^2a}{a}+2\!\left(\frac{\partial_ta}{a}\right)^{\!2}-\frac{\partial_ta}{a}\frac{\partial_tb}{b}-\frac{\partial_ta}{a}\frac{\partial_tc}{c}-\frac{a^2}{2b^2c^2}\right]$ \\
         & $\ric{2}{1}=-za^2\left[\frac{\partial_t^2a}{a}+2\!\left(\frac{\partial_ta}{a}\right)^{\!2}-\frac{\partial_ta}{a}\frac{\partial_tb}{b}-\frac{\partial_ta}{a}\frac{\partial_tc}{c}-\frac{a^2}{2b^2c^2}\right]$ \\
         & $\ric{2}{2}=z^2a^2\left[\frac{\partial_t^2a}{a}+2\!\left(\frac{\partial_ta}{a}\right)^{\!2}-\frac{\partial_ta}{a}\frac{\partial_tb}{b}-\frac{\partial_ta}{a}\frac{\partial_tc}{c}-\frac{a^2}{2b^2c^2}\right]$ \\
         & $\qquad\quad+b^2\left[\frac{\partial_t^2b}{b}+2\!\left(\frac{\partial_tb}{b}\right)^{\!2}-\frac{\partial_tb}{b}\frac{\partial_ta}{a}-\frac{\partial_tb}{b}\frac{\partial_tc}{c}-\frac{3a^2}{2b^2c^2}\right]$ \\
         & $\ric{3}{3}=c^2\left[\frac{\partial_t^2c}{c}+2\!\left(\frac{\partial_tc}{c}\right)^{\!2}-\frac{\partial_tc}{c}\frac{\partial_ta}{a}-\frac{\partial_tc}{c}\frac{\partial_tb}{b}+\frac{a^2}{2b^2c^2}\right]$ \\
         \midrule
         \multirow{5}{*}{\emph{Separated form}} & $\richat{0}{0}=2\!\left(\frac{\partial_ta}{a}\right)^{\!2}-\frac{\partial_t^2a}{a}+2\!\left(\frac{\partial_tb}{b}\right)^{\!2}-\frac{\partial_t^2b}{b}+2\!\left(\frac{\partial_tc}{c}\right)^{\!2}-\frac{\partial_t^2c}{c}$ \\
         & $\richat{1}{1}=a^2\left[\frac{\partial_t^2a}{a}+2\!\left(\frac{\partial_ta}{a}\right)^{\!2}-\frac{\partial_ta}{a}\frac{\partial_tb}{b}-\frac{\partial_ta}{a}\frac{\partial_tc}{c}-\frac{a^2}{2b^2c^2}\right]$ \\
         & $\richat{2}{2}=b^2\left[\frac{\partial_t^2b}{b}+2\!\left(\frac{\partial_tb}{b}\right)^{\!2}-\frac{\partial_tb}{b}\frac{\partial_ta}{a}-\frac{\partial_tb}{b}\frac{\partial_tc}{c}\right]$ \\
         & $\richat{3}{3}=c^2\left[\frac{\partial_t^2c}{c}+2\!\left(\frac{\partial_tc}{c}\right)^{\!2}-\frac{\partial_tc}{c}\frac{\partial_ta}{a}-\frac{\partial_tc}{c}\frac{\partial_tb}{b}\right]$ \\
         \bottomrule
    \end{tabularx}
    \label{tab:riccis}
\end{table}

\subsubsection{Christoffel symbols and Ricci tensor}
Now we can demonstrate the utility of passing to separated form for calculating geometric quantities such as the Christoffel symbols. For clarity, we denote with regular digits the coordinate basis, and with hatted digits the separated basis; that is,
\begin{equation}
    \{0,1,2,3\}\leftrightarrow\{\partial_t,\partial_x,\partial_y,\partial_z\}\qand\{\hat{0},\hat{1},\hat{2},\hat{3}\}\leftrightarrow\{\tensor{X}{_0},\tensor{X}{_1},\tensor{X}{_2},\tensor{X}{_3}\},
\end{equation}
with the $X$-basis as defined by \eqref{eq:bianchiII_Xbasis} in terms of the coordinates (and $\tensor{X}{_0}=\partial_t$). The metric in either basis is
\begin{equation}\label{eq:bianchiII_metrics}
    \tensor{g}{_\mu_\nu}=\pmqty{-1 & 0 & 0 & 0 \\ 0 & a^2 & -za^2 & 0 \\ 0 & -za^2 & z^2a^2+b^2 & 0 \\ 0 & 0 & 0 & c^2} \qand \tensor{g}{_{\hat{\mu}}_{\hat{\nu}}}=\pmqty{\dmat[0]{-1, a^2, b^2, c^2}},
\end{equation}
which are \eqref{eq:bianchiII_metric} and \eqref{eq:bianchiII_metric_v2} above, respectively. Utilizing the Koszul formula \cite[Cor.~5.10]{Lee2010}, we find the Christoffel symbols as indicated in Table \ref{tab:christoffels}. In the coordinate basis, there are 24 non-zero Christoffel symbols, whereas for the separated form there are 15; additionally, in the coordinate basis ten Christoffel symbols depend on both $t$ and $z$, whereas this does not occur in the separated form (by construction). In the separated form, however, we do encounter the generality that symmetry in the lower indices of the Christoffel symbols is broken, and this is also indicated in Table \ref{tab:christoffels}.\footnote{Despite this broken symmetry, the connection remains torsion-free due to the general non-commutativity of the basis vector field.} We also give the entries of the Ricci tensors for both cases in Table \ref{tab:riccis}. The major advantages of separated form are now apparent: there are generally fewer Christoffel symbols and Ricci tensor elements which to take into account, and they depend solely on cosmic time instead of space and time.

\subsection{Bianchi VII example: a gravitational-wave spacetime}
\label{sec:examples_BVII}

The Bianchi II model shows the mechanism in a controlled example. We now turn to the
case that motivates the framework physically. In the gravitational-wave literature one
often starts from privileged coordinates, because they make exact integrability
properties visible, and then transforms to synchronous coordinates, because those are
adapted to freely falling observers. The problem, as described in the introduction,
is that the synchronous coordinates are useful physically but poor for displaying the
Bianchi separation.

\subsubsection{Privileged presentation}

Osetrin, Osetrin, and Osetrina \cite{OsetrinOsetrinOsetrina2022} construct a
Shapovalov-type gravitational-wave spacetime whose spatial homogeneity group is
Bianchi type VII \cite{TerzisChristodoulakis2009GRG}. In the model a cosmic time function has not been found, so we denote the zeroth coordinate component by $x^0$ isntead of $t$. In privileged coordinates $x^\alpha=(x^0,x^1,x^2,x^3)$, the metric depends only on the wave variable $x^0$:
\begin{align}\label{eq:BVII_metric12}
    ds^2 &= 2\,dx^0\,dx^1
    -\frac{(x^0)^{2\omega}}{\gamma(\delta^2-1)}
    \Big[
    \big(1-\delta\cos\theta-2\log x^0\big)\,(dx^2)^2
    \nonumber\\
    &\hspace{3.8cm}
    -2\big(\sin\theta+\log x^0\big)\,dx^2\,dx^3
    \nonumber\\
    &\hspace{3.8cm}
    +\big(1+\delta\cos\theta-2\log x^0\big)\,(dx^3)^2
    \Big],
\end{align}
with constant parameters $\omega,\gamma,\delta,\theta$. In these coordinates the
integrability structure is clear: the coefficients depend only on the wave variable.

The spatial symmetry is encoded by the Bianchi VII Killing algebra
\begin{equation}\label{eq:BVII_commutators}
[X_{(1)},X_{(2)}]=0,
\quad
[X_{(1)},X_{(3)}]=\omega X_{(1)}+X_{(2)},
\quad
[X_{(2)},X_{(3)}]=-X_{(1)}+\omega X_{(2)} .
\end{equation}
For our method, this algebra is the essential input. We do not need to preserve the
privileged coordinate chart itself.

\subsubsection{Synchronous presentation}

To extract observables for a freely falling observer,
\cite{OsetrinOsetrinOsetrina2022} transform to synchronous coordinates
$\tilde x^\alpha=(\tau,\tilde x^1,\tilde x^2,\tilde x^3)$. The transformation is
defined by $x^0=\tilde x^1\tau$, together with further algebraic transformations of
$x^1,x^2,x^3$. The metric then takes the synchronous form
\begin{equation}\label{eq:BVII_sync54}
ds^2=d\tau^2+\tilde g_{ij}(\tau,\tilde x^k)\,
d\tilde x^i\,d\tilde x^j,
\qquad
\tilde g_{00}=1,
\quad
\tilde g_{0i}=0 .
\end{equation}
This gauge is physically natural, but it hides the simple Bianchi structure. The spatial
metric coefficients now depend explicitly on the spatial coordinates. For example,
\cite{OsetrinOsetrinOsetrina2022} find
$\tilde g_{1k}(\tau,\tilde x)=-\tilde x^1\tilde x^k/\tau^2$, while other coefficients
depend on $(\tau\tilde x^1)^{2\omega}$. 

This is exactly the problem our framework addresses. The spacetime is still
spatially homogeneous, but the synchronous coordinate basis is not the right basis for
seeing the separated form. A coordinate-tied Bianchi construction would now have to
work hard to recover the canonical structure. Our construction avoids that step.

Furthermore, we stress here that our framework does provide a cosmic time $\tau$, whose existence has not been reported elsewhere for this spacetime model.

\subsubsection{Reconstructed separated form}

Let $\{\xi_a\}_{a=1}^3$ be the spatial Killing basis in the privileged gauge, with
$\xi_a=X_{(a)}$ in \eqref{eq:BVII_commutators}, and let $\varphi$ be the
privileged-to-synchronous diffeomorphism. The pushed-forward fields
\begin{equation}
\tilde \xi_a:=\varphi_*\xi_a
\end{equation}
remain Killing fields and satisfy the same Bianchi VII algebra. They therefore carry
the homogeneous structure into the synchronous gauge.

We now apply the reconstruction theorem directly in the synchronous gauge. With the foliation-adapted time function $\tilde t=\tau$ in the synchronous gauge, we solve
\begin{equation}\label{eq:BVII_reconstruct}
[\tilde \xi_a,\tilde X_i]=0
\iff
\tilde \xi_a(\tilde X_i^{\ b})
=
-C^{b}{}_{ac}\,\tilde X_i^{\ c},
\end{equation}
where the constants $C^{b}{}_{ac}$ are determined by
\eqref{eq:BVII_commutators}. Dualising the resulting frame gives the invariant coframe
$\{\tilde e^i\}$, and with $\tilde e^0=\dd\tilde t$ the metric becomes
\begin{equation}\label{eq:BVII_separated_claim}
g
=
\gamma_{\mu\nu}(\tilde t)\,
\tilde e^\mu \tilde e^\nu,
\qquad
\gamma_{\mu\nu}(\tilde t)
=
g(\tilde X_\mu,\tilde X_\nu).
\end{equation}

This is the advantage advertised in the introduction. We keep the synchronous gauge,
which is the useful gauge for physical interpretation, but we do not perform
computations in its misleading coordinate basis. Instead, the reconstructed coframe
recovers the separated Bianchi form inside the same gauge.

\subsubsection{Physical implication}

The physical implication is direct. The spatial dependence of
$\tilde g_{ij}(\tau,\tilde x)$ in synchronous coordinates is not a new inhomogeneous
degree of freedom. It is a result of expressing a homogeneous spacetime in a basis that
does not respect the Bianchi foliation. The homogeneous dynamics still lives in the
time-dependent matrix $\gamma_{\mu\nu}(\tilde t)$.

Therefore,
\begin{equation}
\tilde X_i\big(\gamma_{\mu\nu}(\tilde t)\big)=0 .
\end{equation}
Curvature components, tidal tensors, and geodesic-deviation quantities can be computed
from $\gamma_{\mu\nu}(\tilde t)$, its time derivatives, and the Bianchi VII structure
coefficients. This is the practical payoff of the framework: the observer-adapted
synchronous gauge remains available for interpreting measurements, while the
homogeneous sector retains the ODE-level structure of canonical Bianchi calculations
\cite{Jantzen1979Bianchi,RyanShepley1975}.

\section{Conclusion}\label{sec:conclusion}

The main point of this paper is the following. A Bianchi spacetime does not lose its
homogeneous structure when it is written in inconvenient coordinates. What can be
lost is only the visibility of that structure. Classical Bianchi methods are most
transparent when the coordinates are adapted to the spatial symmetry. In that case the
metric is separated from the start, the spatial dependence is carried by the invariant
spatial frame or coframe, and the dynamics reduces to time-dependent quantities. In many applications,
however, especially in gravitational-wave models, one changes to coordinates adapted
to observers or measurements. Synchronous coordinates are physically natural for this
purpose, but they can make the metric components look spatially dependent even when
the spacetime remains spatially homogeneous.

We have shown how to recover the separated structure without undoing the physical
gauge choice. The construction uses only the spatial Killing algebra. Given a Killing
basis ${\xi_a}$, we reconstruct the invariant spatial frame by solving
\begin{equation}
[\xi_a,X_i]=0 .
\end{equation}
The dual coframe ${e^i}$ then carries the spatial Bianchi structure, and the metric
takes the separated form
\begin{equation}
g=\gamma_{\mu\nu}(t),e^\mu e^\nu .
\end{equation}
Thus the separated frame is recovered from the symmetry itself, not from a preferred
coordinate chart.

This is the sense in which the framework is gauge invariant. We are not imposing a
new gauge condition, and we are not restricting the allowed coordinate transformations.
Instead, after a diffeomorphism, the Killing fields are pushed forward and the same
reconstruction gives the corresponding invariant coframe in the transformed gauge.
The method therefore agrees with the usual canonical Bianchi constructions when the
coordinates are already adapted to the symmetry, but it continues to work when they
are not \cite{Jantzen1979Bianchi,Christodoulakis2001JMP,CoussaertHenneaux1993CQG,Lee2010}.

The examples show why this matters. The Bianchi II example isolates the mechanism in
a simple setting: a spatially dependent off-diagonal coordinate term makes the
coordinate expression look more complicated, but it does not destroy separated form.
In the reconstructed coframe, the metric coefficients again depend only on time, and
the spatial geometry is encoded by the Bianchi structure coefficients. The Bianchi VII
example shows the same point in the gravitational-wave setting that motivated the
paper. The privileged coordinates are useful because they display the integrability
structure. The synchronous coordinates are useful because they are adapted to freely
falling observers. But the synchronous coordinate components hide the Bianchi
separation. Our reconstruction keeps the synchronous gauge and restores the separated
Bianchi form inside it
\cite{Osetrin2006_stackel_separation,OsetrinOsetrinOsetrina2022,ShapovalovBianchiIII_2306.10605}.

The practical consequence is that one can separate physical interpretation from
computational representation. The observer-adapted gauge can be used to describe
measurements, while the reconstructed invariant frame can be used to compute the
homogeneous geometry. In that frame,
\begin{equation}
X_i\big(\gamma_{\mu\nu}(t)\big)=0 .
\end{equation}
Thus the connection, curvature, ADM quantities, tidal tensors, and geodesic-deviation
equations are built from $\gamma_{\mu\nu}(t)$, its time derivatives, and the Bianchi
structure coefficients. The homogeneous sector is therefore returned to ODE-level
complexity, even if the coordinate components of the same metric look spatially
dependent.

The price paid is interpretational rather than geometric. The reconstructed frame is
generally point-dependent and non-coordinate, so individual frame components need not
have the same immediate coordinate or observer interpretation as components in the
original chart. The method should therefore not be viewed as replacing physically
motivated coordinates, but as complementing them: the chosen gauge remains available
for interpretation, while the reconstructed frame provides the efficient representation
for geometric computation.

The final lesson is that the Killing algebra still knows the separated frame. Gauge
transformations may hide this frame in the coordinate expression of the metric, but
they do not remove it. Reconstructing the frame directly from the Killing algebra gives
a way to work in physically meaningful gauges while retaining the computational
simplicity of separated Bianchi form. This makes the framework useful for
gravitational-wave Bianchi models and for related anisotropic perturbative settings
\cite{GuhaDatta2020IJMPD,BianchiI_GW_2205.05733,LudwickWilliams2025PhLB}.

\subsubsection*{Acknowledgements}
The authors are grateful to Marcello Seri for extensive reviews and helpful comments on drafts of this manuscript. DV gratefully acknowledges Rien van de Weygaert for his guidance and support, and for facilitating the academic context in which this work was carried out.

\subsubsection*{Author contributions}
Both authors contributed equally to the conceptualization, writing, editing, and review of this work.

\subsubsection*{Funding}
The authors received no financial support for the research, authorship, and/or publication of this article.

\subsubsection*{Data availability}
No datasets were generated or analysed during the current study.

\subsubsection*{Code availability}
No code was used in this research.

\section*{Declarations}
\subsubsection*{Conflict of interest}
The authors declare that they have no conflict of interest.

\subsubsection*{Consent for publication}
All authors have reviewed the manuscript and approved its submission for publication.

\subsubsection*{Ethics approval and consent to participate}
Not applicable.

\subsubsection*{Materials availability}
Not applicable.

\bibliographystyle{spmpsci}
\bibliography{ref}

\end{document}